\documentclass[11pt,letterpaper]{article}

\usepackage[T1]{fontenc}
\usepackage[utf8]{inputenc}
\usepackage{lmodern,microtype}
\usepackage[margin=1in,headheight=14pt]{geometry}
\usepackage{amsmath,amssymb,amsthm,mathtools,mathrsfs}
\usepackage{enumitem}
\usepackage[numbers,sort&compress]{natbib}
\usepackage{xcolor,fancyhdr}
\usepackage{hyperref}
\definecolor{linkblue}{HTML}{1B4058}
\hypersetup{colorlinks=true,linkcolor=linkblue,citecolor=linkblue,urlcolor=linkblue}

\allowdisplaybreaks[1]
\numberwithin{equation}{section}
\newtheorem{theorem}{Theorem}[section]
\newtheorem{proposition}[theorem]{Proposition}

\newtheorem{corollary}[theorem]{Corollary}
\theoremstyle{definition}
\newtheorem{definition}[theorem]{Definition}

\theoremstyle{remark}
\newtheorem{remark}[theorem]{Remark}
\DeclareMathOperator{\Tr}{Tr}

\newcommand{\F}{\mathbb F}
\newcommand{\Z}{\mathbb Z}
\newcommand{\V}{V^n}
\newcommand{\C}{\mathbb C}
\newcommand{\HH}{\mathcal H}
\newcommand{\Hn}{\HH^{\otimes n}}
\newcommand{\LL}{\mathcal L}
\newcommand{\conv}{\boxtimes}
\newcommand{\Rz}{Rz}
\newcommand{\ket}[1]{|#1\rangle}
\newcommand{\bra}[1]{\langle#1|}
\newcommand{\proj}[1]{|#1\rangle\langle#1|}

\newcommand{\CA}{\mathcal C}
\newcommand{\cF}{\mathscr F}
\newcommand{\cI}{\mathcal I}
\newcommand{\Ad}{\mathscr A}

\setlist[enumerate]{label=\textup{(\roman*)},leftmargin=2em,itemsep=3pt,topsep=4pt}

\title{\bfseries Submodularity of entropy under quantum convolution}

\author{Milad M. Goodarzi\\[10pt]
\small Centre for Quantum Technologies, National University of Singapore, Singapore 117543, Singapore\\[2pt]
\small \texttt{milad.moazami@gmail.com}
}
\date{}

\begin{document}
\maketitle
\thispagestyle{plain}
\begin{abstract}
We develop a submodular framework for the von Neumann entropy of discrete quantum convolutions, providing a noncommutative counterpart to the direct side of entropic additive combinatorics. We first introduce globally weighted quantum convolutions, which form compatible families indexed by admissible subsets of a fixed collection of inputs. Our main theorem reveals a polymatroidal geometry underlying their entropy growth: relative to any fixed admissible input block, the entropy gains admit a normalized, monotone, submodular extension to all subsets of the remaining inputs. The theorem yields convolutional strong subadditivity, quantum Ruzsa triangle inequality, and quantum entropic Pl\"unnecke--Ruzsa inequalities for arbitrary input states. For repeated inputs, it gives sharp comparisons of entropy growth across admissible scales; in particular, the quantum doubling constant controls all higher admissible convolution entropies with optimal exponents. Together, these results bring submodular methods from additive combinatorics into the quantum setting and provide a systematic route to broad families of convolutional entropy inequalities.
\end{abstract}

\section{Introduction}

A central problem in additive combinatorics is to understand growth under addition. For finite subsets of an abelian group, this growth is reflected in the size and structure of their sumsets \cite{TaoVu}. For independent random variables, the corresponding quantity is the increase of Shannon entropy under convolution. This analogy has developed into the entropy method in additive combinatorics, which seeks probabilistic counterparts of sumset inequalities and uses the submodular structure of entropy to organize additive growth. This line of work was initiated by Ruzsa \cite{Ruzsa}, developed systematically by Tao \cite{Tao}, and placed within a broad framework of submodularity and fractional coverings by Madiman, Marcus, and Tetali \cite{MMT}. More recently, Green, Manners, and Tao \cite{GMT} further developed the theory of entropic doubling and used it to derive new structural results for finite sets with small doubling.

A prototypical instance of this structure is the inequality
\begin{equation}\label{intro-classical-ssa}
 H(X+Y+Z)+H(Y)\le H(X+Y)+H(Y+Z),
\end{equation}
valid for independent, finitely supported random variables taking values in an abelian group. It expresses a diminishing-returns
principle: the entropy gained by adding $X$ decreases when $Y$ has already been combined with the independent input $Z$. More generally,
for independent $X_0,X_1,\ldots,X_r$, the entropy-gain function
\begin{equation}\label{intro-classical-gain}
 J\longmapsto
 H\left(X_0+\sum_{j\in J}X_j\right)-H(X_0)
\end{equation}
is normalized, monotone, and submodular. As developed in \cite{MMT}, submodularity and fractional coverings organize a broad family of
entropy inequalities for sums. Applied to the function above, they convert bounds for smaller collections of summands into bounds for the full sum.

In this paper, we establish a noncommutative counterpart of this structure for discrete quantum convolution. A quantum analogue of \eqref{intro-classical-gain} requires a compatible family of convolution states indexed by subsets, rather than a single binary convolution. We construct such families by introducing globally weighted quantum convolutions. Our starting point is the binary convolution of Bu, Gu, and Jaffe
\cite{BGJPNAS,BGJGaussians}. For an odd prime $d$, their convolution $\rho\conv_{s,t}\sigma$ combines two $n$-qudit states by a unitary followed by a partial trace, where $s,t\in\Z_d^\times$ satisfy $s^2+t^2=1$. Its characteristic function obeys the multiplication rule $\Xi_{\rho\conv_{s,t}\sigma}(x) = \Xi_\rho(sx)\Xi_\sigma(tx)$, analogous to the corresponding identity for classical convolution. Guided by this identity, we assign a fixed global weight to each input and normalize the weights separately on every admissible subset. This produces a compatible family of multi-input convolutions that extends the binary construction and allows convolution entropies associated with different subsets to be compared. Our construction also applies when a multi-input convolution cannot be obtained through a sequence of admissible binary convolutions.

The entropy structure of these families is far from immediate. Coherent mixing can create entanglement, and discarding one output can increase entropy even when all inputs are pure. Consequently, the classical bound $H(X+Y)\le H(X)+H(Y)$ has no general counterpart for the input von Neumann entropies; for example, complementary pure stabilizer inputs can have a maximally mixed convolution. Moreover, ordinary strong subadditivity applied directly to the physical outputs does not compare the different convolution states that occur in the analogue of \eqref{intro-classical-ssa}. This difficulty already appears in two conjectures of Bu, Gu, and Jaffe \cite{BGJRuzsa}: the triangle inequality for quantum Ruzsa divergence and convolutional strong
subadditivity \cite[Conjectures 1 and 2]{BGJRuzsa}. They verified the latter, and hence the former, for stabilizer inputs and for states diagonal in the computational basis. These conjectures are the simplest three-input instances of the broader question addressed here: whether the entropy gains of compatible quantum convolutions inherit the submodular structure of classical convolution entropy.

\subsection{Main results}

Our main contributions are as follows.
\begin{itemize}
    \item The first result is a globally weighted construction of quantum convolutions (Theorem~\ref{product theorem} and Definition~\ref{compatible}). We assign a fixed global weight to each input and, for every admissible subset \(A\), normalize the corresponding weights to obtain a convolution state \(\CA_A\). The resulting family is compatible across admissible subsets. This construction extends the binary framework of \cite{BGJPNAS,BGJGaussians} and includes multi-input convolutions that cannot be obtained by successively applying admissible binary convolutions. We also realize each \(\CA_A\) as the output of a quantum channel.

    \item Our main result (Theorem~\ref{main theorem}) shows that the entropy gains of every such compatible family are governed by a polymatroidal geometry. Relative to any fixed nonempty admissible block \(R\), the physical entropy gains extend to a normalized, monotone, submodular function on the full subset lattice of the remaining inputs. This extension is defined even at subsets $J$ for which the physical convolution $\CA_{R\cup J}$ is not. Fractional subadditivity then yields a general hierarchy of convolutional entropy inequalities. For example, if $M = \{i_1,\ldots,i_r\}$, then, whenever the convolutions appearing below are admissible,
    \begin{equation}
        S(\CA_{R\cup M}) - S(\CA_R) \leq \sum_{j = 1}^{r} S(\CA_{R\cup\{i_j\}}) - S(\CA_R).
    \end{equation}
    Thus the entropy gain from adding all $r$ inputs is bounded by the sum of the gains obtained by adding them individually.

    \item As three-input consequences, we obtain the convolutional strong subadditivity (Corollary~\ref{corollary weighted-ssa})
    \begin{equation}
        S\big((\rho\conv_{s,s}\tau)\conv_{l,m}\sigma\big) + S(\sigma) \leq S(\rho\conv_{s,s}\sigma) + S(\sigma\conv_{s,s}\tau),
    \end{equation}
    whenever $2s^2 = 1$, $l^2 + m^2 = 1$, $m = ls$, and, the quantum Ruzsa triangle inequality (Corollary~\ref{corollary ruzsa})
    \begin{equation}
        S(\rho\conv_{s,s}\tau) + S(\sigma) \leq S(\rho\conv_{s,s}\sigma) + S(\sigma\conv_{s,s}\tau),
    \end{equation}
    under $2s^2 = 1$. These results resolve Conjectures~2 and~1, respectively, of \cite{BGJRuzsa} for arbitrary quantum states.

    \item At higher input counts, the same structure yields quantum entropic Pl\"unnecke--Ruzsa inequalities and sharp comparisons between repeated convolutions at different admissible scales (Corollaries~\ref{corollary four-input} and~\ref{corollary balanced-growth}). In particular, the entropy gain per added input is nonincreasing along admissible repetition scales. Equivalently, the multiplicative entropy growth of an admissible $m$-fold convolution is at most $\delta_q[\rho]^{m-1}$, where $\delta_q[\rho]$ is the quantum doubling constant. The exponent $m-1$ is optimal, already among states diagonal in the computational basis.
\end{itemize}

These results give a common mechanism for quantum entropy-growth inequalities across different weights, subsets, and repetition scales. The global formulation makes the arithmetic requirements explicit while allowing the entropy argument to proceed on the entire subset lattice.

\subsection{Proof method}

Our proof of the main result is based on a representation that connects convolution entropy with the marginal entropies of a single auxiliary state.
Inspired in part by the strong dynamical subadditivity of Roga, Fannes, and \.Zyczkowski \cite{RFZ}, we seek an entropy-preserving representation in which combining inputs becomes a tractable matrix operation.

Let $\mathcal H=(\mathbb C^d)^{\otimes n}$, let $D=d^n$, and let $\V=\Z_d^n\times\Z_d^n$ be the discrete phase space, of cardinality $N=D^2$. For $x\in\V$, denote by $w(x)$ the corresponding discrete Weyl operator. The operators $\{w(x):x\in\V\}$ form an orthogonal projective representation of the phase space. We associate with every state $\rho$ on $\mathcal H$ the positive characteristic kernel $K_\rho(x,y)=\Tr[\rho w(x)^\dagger w(y)]$. This is the positivity kernel appearing in the quantum Bochner theorem \cite{Bochner}. Two identities provide the bridge:
\begin{equation}\label{intro-bridge}
    S(K_\rho/N) = S(\rho) + \log D, \qquad K_{\rho\conv_{s,t}\sigma} = K_\rho^{[s]}\odot K_\sigma^{[t]},
\end{equation}
where $K^{[a]}(x,y)=K(ax,ay)$ and $\odot$ denotes entrywise multiplication. The first identity identifies the state entropy up to a universal additive constant; the second turns quantum convolution into a Hadamard product. For globally weighted convolutions, the latter becomes
\begin{equation}\label{intro-weighted-bridge}
    K_{\CA_A}^{[c_A]} = \bigodot_{i\in A}K_{\rho_i}^{[q_i]}.
\end{equation}
We then realize all these kernel products simultaneously through the marginals of one separable state $\Omega$. For every admissible $A$, their spectra give $S(\Omega_A)=S(\CA_A)+\log D$. The auxiliary state has marginals for every subset, including those for which a convolution normalization does not exist. Ordinary strong subadditivity makes $A\mapsto S(\Omega_A)$ submodular, and separability makes it monotone under inclusion. Subtracting the entropy of a fixed block cancels the common $\log D$ term and produces the extension in our main theorem. This simultaneous representation connects the individual convolution outputs to the full family of entropy inequalities.

\paragraph{Organization of the paper.}
In Section~\ref{prelim section}, we fix the Weyl conventions and recall discrete quantum convolution. In Section~\ref{kernel section}, we develop the characteristic-kernel method and the existence and channel realization of globally weighted convolutions. We prove the submodular extension theorem in Section~\ref{main section}. In Section~\ref{corollary section}, we derive convolutional strong subadditivity, the Ruzsa triangle inequality, fractional-cover growth bounds, and the sharp repeated-input estimates. Finally, in Section~\ref{conclusion section} we discuss the scope of these results and the additional ingredients needed for conditional, quotient, and inverse theories in quantum additive combinatorics.

%%%%%%%%%%%%%%%%%%%%%%%%%%%%%%%%%%%%%%%%
%%%%%%%%%%%%%%%%%%%%%%%%%%%%%%%%%%%%%%%%
%%%%%%%%%%%%%%%%%%%%%%%%%%%%%%%%%%%%%%%%

\section{Preliminaries}\label{prelim section}

We use the phase-space and characteristic-function conventions of Bu, Gu, and Jaffe \cite{BGJGaussians,BGJRuzsa}, and follow \cite{BGJRuzsa} for the notation and definitions concerning quantum convolution and Ruzsa-type quantities. Throughout, $d$ is an odd prime and $n \geq 1$. The one-qudit space is $\HH = \C^d$, and the $n$-qudit system is $\Hn$. We use
\begin{equation}
 \V = \Z_d^n\times\Z_d^n, \qquad D=d^n, \qquad N = |\V| = D^2.
\end{equation}
Here $\Z_d = \mathbb Z/d\mathbb Z$, identified with the field $\F_d$. All phase-space and convolution-coefficient arithmetic is in this field.

\subsection{Weyl operators and characteristic functions}

Put $\omega_d = e^{2\pi i/d}$ and $\zeta = \omega_d^{(d+1)/2}$. On the computational basis $\{\ket{j}:j\in\Z_d\}$ of $\HH$, let $X\ket{j} = \ket{j+1}$ and $Z\ket{j} = \omega_d^j\ket{j}$. The local Weyl operator is $w(p,q) = \zeta^{-pq}Z^pX^q$.
For the $n$-qudit system, write $x=(p,q)$ for $(\vec p,\vec q)\in\V$, suppressing arrows on phase-space vectors. Then
\begin{equation}\label{weyl}
    w(x) = \bigotimes_{j=1}^n w(p_j,q_j) = \omega_d^{-p\cdot q/2} Z^p X^q, \qquad [x,y] = p\cdot q' - q\cdot p'.
\end{equation}
Here $y = (p',q')$, the powers denote tensor products over the $n$ coordinates, and $1/2$ is the inverse of $2$ in $\Z_d$. Direct calculation gives
\begin{align}
    w(x)w(y) &= \omega_d^{[x,y]/2}w(x+y), &w(x)^\dagger &= w(-x),\label{weyl-product}\\
    w(x)^\dagger w(y) &= \omega_d^{-[x,y]/2}w(y-x), &\Tr[w(x)^\dagger w(y)] &= D \delta_{xy}.
\end{align}
The symbol $L^2(\Hn)$ denotes the Hilbert space of linear operators acting on $\Hn$ with the normalized Hilbert--Schmidt inner product $\langle A,B\rangle = D^{-1}\Tr(A^\dagger B)$. Then $\{w(x):x\in\V\}$ is an orthonormal basis. With the unnormalized trace inner product, the normalized basis is instead $\{w(x)/\sqrt{D}:x\in\V\}$.

We use the characteristic-function convention common to \cite{BGJPNAS,BGJGaussians,BGJRuzsa}:
\begin{equation}\label{characteristic}
    \Xi_\rho(x) = \Tr[\rho w(-x)], \qquad \rho = \frac{1}{D} \sum_{x\in \V} \Xi_\rho(x)w(x).
\end{equation}
This is a Fourier transform in an operator basis. The underlying phase space is abelian, but its Weyl representation is projective, as the phase in \eqref{weyl-product} records. Consequently, the ordinary positive-definiteness criterion for a classical characteristic function must be replaced by one that includes this phase; see \eqref{bochner} below.

\subsection{Discrete quantum convolution}

Following \cite[Definition 10]{BGJRuzsa}, let $s,t \in \Z_d^\times$ satisfy $s^2 + t^2 = 1$; we call such $(s,t)$ an admissible pair. The corresponding discrete beam-splitter unitary is defined by
\begin{equation}\label{beam-splitter}
    U_{s,t}\ket{i,j} = \ket{si + tj,-ti + sj}, \qquad i,j \in\Z_d^n.
\end{equation}
The quantum convolution of states $\rho$ and $\sigma$ is the state obtained by applying this unitary and discarding the second output:
\begin{equation}\label{binary-convolution}
    \rho\conv_{s,t}\sigma = \Tr_B\!\left[U_{s,t}(\rho\otimes\sigma)U_{s,t}^\dagger\right].
\end{equation}
We sometimes write $\conv$ when the chosen pair is fixed; the balanced case $s =t$ exists precisely when $2$ is a nonzero square in $\mathbb{Z}_d$. The characteristic function identity is
\begin{equation}\label{binary-characteristic}
 \Xi_{\rho\conv_{s,t}\sigma}(x)
 =\Xi_\rho(sx)\Xi_\sigma(tx)
 \quad\text{\cite[Lemma 11(1)]{BGJRuzsa}},
\end{equation}
which immediately implies $\rho\conv_{s,t}\sigma = \sigma\conv_{t,s}\rho$. Since convolution with a fixed coefficient pair is not assumed to be associative, all convolutions involving more than two inputs will be given with their coefficients and bracketing explicitly. In particular, the fixed-pair iteration of \cite[Lemma 11(3)]{BGJRuzsa} is
\begin{equation}\label{bgj-iteration}
    \conv_{s,t}^{0}\rho = \rho, \qquad \conv_{s,t}^{r+1}\rho = (\conv_{s,t}^{r}\rho)\conv_{s,t}\rho, \qquad r \geq 0.
\end{equation}
This iteration should be distinguished from the equal-weight multi-input convolution introduced in Subsection~\ref{repeated inputs and optimal coefficients}, for which we use different notation.

%%%%%%%%%%%%%%%%%%%%%%%%%%%%%%%%%%%%%
%%%%%%%%%%%%%%%%%%%%%%%%%%%%%%%%%%%%%
%%%%%%%%%%%%%%%%%%%%%%%%%%%%%%%%%%%%%

\section{The characteristic-kernel method}\label{kernel section}

Our proof of the main result rests on two observations. First, quantum convolution becomes entrywise multiplication of characteristic kernels (Theorem~\ref{product theorem}). Second, the normalized characteristic kernel has the same nonzero spectrum as the underlying state, up to a fixed multiplicity, so its entropy differs from the state entropy only by $\log D$ (Proposition~\ref{proposition spectrum}). Together, these facts allow ordinary entropy inequalities for a suitable auxiliary state to be transferred to quantum convolution.

\begin{definition}[Characteristic kernel]
For an operator $A$ on $\Hn$, its characteristic kernel is the $N\times N$ matrix
\begin{equation}\label{kernel}
    K_A(x,y) = \Tr[A w(x)^\dagger w(y)] = \omega_d^{-[x,y]/2} \Xi_A(x-y), \qquad x,y\in\V.
\end{equation}
\end{definition}
The phase in \eqref{kernel} is part of the definition. In general, the matrix $\Xi_\rho(x-y)$ without that phase is not positive.

We recall the finite-dimensional Bochner theorem in our conventions \cite{Bochner}: a function $f:\V\to\C$ is the characteristic function of a state on $\Hn$ if and only if $f(0) = 1$ and
\begin{equation}\label{bochner}
    \left[\omega_d^{-[x,y]/2}f(x-y)\right]_{x,y\in\V} \geq 0.
\end{equation}
The corresponding state is uniquely determined by $\rho_f = \frac{1}{D}\sum_{x\in\V}f(x)w(x)$.

\subsection{Globally weighted convolutions}

For a matrix indexed by $\V$ and a scalar $c\in\Z_d^\times$, let
\begin{equation}\label{reindexing}
    K^{[c]}(x,y) = K(cx,cy).
\end{equation}
This is simultaneous permutation of rows and columns, hence preserves the spectrum. For two matrices of the same size, write $(K\odot L)(x,y) = K(x,y)L(x,y)$.

\begin{theorem}\label{product theorem}
Let $\cI$ be a finite nonempty index set, let $\rho_i$ be states, and let $q_i\in\Z_d^\times$ for $i\in\cI$. If a nonempty $A \subseteq \cI$ satisfies $\sum_{i\in A} q_i^2 = c_A^2$, for some $c_A\in\Z_d^\times$, then, for each choice of the square root $c_A$, there exists a unique state $\CA_A$ such that
\begin{equation}\label{multi-kernel}
    K_{\CA_A}^{[c_A]} = \bigodot_{i\in A}K_{\rho_i}^{[q_i]}.
\end{equation}
In particular, for every admissible binary pair $(a,b)$, $K_{\rho\conv_{a,b}\sigma} = K_\rho^{[a]}\odot K_\sigma^{[b]}$. Furthermore, $\mathcal{C}_A$ admits a channel realization, and changing the sign of $c_A$ changes $\CA_A$ by parity and leaves its entropy unchanged.
\end{theorem}
\begin{proof}
The binary case follows directly from the characteristic function identity \eqref{binary-characteristic}:
\begin{align}
    K_\rho(ax,ay)K_\sigma(bx,by) &= \omega_d^{-(a^2+b^2)[x,y]/2}\Xi_\rho(a(x-y))\Xi_\sigma(b(x-y))\\
    &= \omega_d^{-[x,y]/2}\Xi_{\rho\conv_{a,b}\sigma}(x-y),
\end{align}
as desired.

For the multi-input case, put $b_i = q_i/c_A$ for $i\in A$. Then $\sum_i b_i^2 = 1$. The Schur product theorem shows that $\bigodot_i K_{\rho_i}^{[b_i]}$ is positive semidefinite; its entries are
\begin{equation}
    \omega_d^{-[x,y]/2}\prod_{i\in A} \Xi_{\rho_i}(b_i(x-y)).
\end{equation}
The function in the product equals one at zero. The Bochner criterion \eqref{bochner} therefore proves existence of a state $\mathcal{C}_A$ with characteristic function
\begin{equation}\label{multi-characteristic}
    \Xi_{\CA_A}(x) = \prod_{i\in A} \Xi_{\rho_i}\left(\frac{q_i}{c_A}x\right),
\end{equation}
and Fourier inversion proves uniqueness. Reindexing by $c_A$ gives \eqref{multi-kernel}. Changing $c_A$ to $-c_A$ composes the output with the parity unitary $\ket{x}\mapsto\ket{-x}$, and hence does not change its entropy.

\medskip
For coefficients $b_i\in\Z_d$ with $\sum_{i\in A} b_i^2 = 1$, the rule
\begin{equation}\label{weyl-embedding}
    \Phi_A^*(w(x)) = \bigotimes_{i\in A}w(b_ix)
\end{equation}
extends linearly to a unital $*$-homomorphism from the algebra of operators on $\Hn$ into the algebra of operators on $(\Hn)^{\otimes |A|}$. To see this, note that the Weyl operators are a basis, so \eqref{weyl-embedding} defines a linear map. It preserves the identity and adjoints. Moreover,
\begin{align}
    \Phi_A^*(w(x))\Phi_A^*(w(y)) &= \omega_d^{[x,y]\sum_i b_i^2/2}\bigotimes_i w(b_i(x+y))\\
    &= \Phi_A^*(w(x)w(y)).
\end{align}
Bilinearity gives multiplicativity on all operators. A $*$-homomorphism is completely positive, so its trace adjoint is completely positive and trace preserving.

We now identify its action on product inputs. For every $x\in\V$,
\begin{align}
    \Xi_{\Phi_A\left(\bigotimes_{i\in A}\rho_i\right)}(x) &= \Tr\!\left[\Phi_A\left(\bigotimes_{i\in A}\rho_i\right)w(-x)\right]\\
    &= \Tr\!\left[\left(\bigotimes_{i\in A}\rho_i\right)\Phi_A^*(w(-x))\right]\\
    &= \prod_{i\in A} \Tr[\rho_iw(-b_ix)] = \prod_{i\in A} \Xi_{\rho_i}(b_ix).
\end{align}
For $b_i = q_i/c_A$, the final expression is precisely $\Xi_{\CA_A}(x)$ by \eqref{multi-characteristic}. Since the characteristic function uniquely determines the state, it follows that
\begin{equation}
    \Phi_A\left(\bigotimes_{i\in A}\rho_i\right) = \CA_A.
\end{equation}
Equivalently, a finite-dimensional representation of the full matrix algebra is unitarily equivalent to $A\mapsto A\otimes I$ with the appropriate multiplicity. Thus this channel can be realized by a unitary followed by a partial trace.
\end{proof}

\begin{remark}[Harmonic-analytic interpretation]
The finite-dimensional harmonic analysis also explains the normalization conditions. The Weyl cocycle supplies the phase in the characteristic kernel, and Hadamard multiplication adds the phase coefficients associated with the squared input weights. The same identity therefore accounts for both the convolution rule and the arithmetic constraints on admissible subsets.
\end{remark}

\begin{remark}
If $B,C$ are disjoint and $B,C,B\cup C\in\Ad$, then
\begin{equation}\label{grouping}
    \CA_{B\cup C} = \CA_B\conv_{c_B/c_{B\cup C},\,c_C/c_{B\cup C}}\CA_C.
\end{equation}
Indeed, the two displayed binary coefficients are nonzero and their squares sum to one. Applying \eqref{binary-characteristic} and then \eqref{multi-characteristic} yields the same characteristic function on both sides.
\end{remark}

Theorem~\ref{product theorem} motivates the following

\begin{definition}[Globally weighted convolutions]\label{compatible}
Let $\cI$ be a finite nonempty index set, let $\rho_i$ be states, and let $q_i\in\Z_d^\times$ for $i\in\cI$. Let $\Ad$ denote the collection of nonempty subsets $A$ satisfying the admissibility condition $\sum_{i\in A} q_i^2 = c_A^2$, for some $c_A\in\Z_d^\times$. For each $A\in\Ad$, we call the state $\CA_A$ of Theorem~\ref{product theorem} the \textit{globally weighted convolution} associated with the inputs and weights indexed by $A$. We call the family $\{\CA_A:A\in\Ad\}$ the compatible family of globally weighted convolutions associated with $(q_i)_{i\in\cI}$.
\end{definition}

For each subset $A$, either choice of square root ($c_A$ or $-c_A$) may be used independently; the resulting entropy quantities are unchanged. A singleton is always admissible, and $S(\CA_{\{i\}}) = S(\rho_i)$. The collection $\Ad$ need not be closed under unions or intersections. In particular, no unspecified convolution is assigned to a subset whose squared-weight sum is zero or is a nonsquare.

\subsection{Entropy of the normalized kernel}

\begin{proposition}[Spectrum and entropy of the characteristic kernel]\label{proposition spectrum}
Let $\rho$ be a quantum state on $\Hn$, and let $K_\rho$ be its characteristic kernel. If $\lambda_1,\ldots,\lambda_D$ are the eigenvalues of $\rho$,
counted with multiplicity, then the eigenvalues of the normalized kernel $\frac{1}{N} K_\rho$ are $\frac{\lambda_1}{D},\ldots,\frac{\lambda_D}{D}$, each repeated $D$ times. In particular, $\frac{1}{N} K_\rho$ is a density matrix on $\ell^2(\V)$ and
\begin{equation}\label{spectral-entropy}
    S\left(\frac{1}{N} K_\rho\right) = S(\rho)+\log D.
\end{equation}
\end{proposition}
\begin{proof}
For a state $\rho$, define vectors in $L^2(\Hn)$ by
\begin{equation}
    a_x^\rho := w(x)\rho^{1/2}, \qquad x\in\V.
\end{equation}
Then, evidently
\begin{equation}\label{kernel Hilbert Schmidt}
    \langle a_x^\rho,a_y^\rho\rangle = \frac{1}{D} \Tr\!\left[\rho\,w(x)^\dagger w(y)\right] = \frac{1}{D} K_\rho(x,y).
\end{equation}
Let $\{e_x:x\in\V\}$ be the canonical orthonormal basis of $\ell^2(\V)$, and define
\begin{equation}
    F:\ell^2(\V) \longrightarrow L^2(\Hn), \qquad \text{by} \quad Fe_x = \frac{1}{\sqrt{D}}\,a_x^\rho.
\end{equation}
Equivalently, using Dirac notation for vectors in the Hilbert--Schmidt space, $F = \frac{1}{\sqrt D} \sum_{x\in\V} \ket{a_x^\rho}\!\bra{e_x}$. It follows from \eqref{kernel Hilbert Schmidt} and $N = D^2$ that
\begin{equation}
    \langle e_x,F^\dagger F e_y\rangle = \langle Fe_x,Fe_y\rangle = \frac{1}{D} \langle a_x^\rho,a_y^\rho\rangle = \frac{1}{N} K_\rho(x,y).
\end{equation}
Hence
\begin{equation}\label{FdaggerF-kernel}
    F^\dagger F = \frac{1}{N} K_\rho.
\end{equation}
On the other hand, for $A\in\LL(\Hn)$, we have
\begin{align}
    FF^\dagger(A) &= \sum_{x\in\V} Fe_x\,\langle Fe_x,A\rangle\\
    &= \frac{1}{D^2} \sum_{x\in\V} \Tr\!\left[A\rho^{1/2}w(x)^\dagger\right]w(x)\rho^{1/2}\\
    &= \frac{1}{D} \left(\frac1{D} \sum_{x\in\V} \Xi_{A\rho^{1/2}}(x) w(x)\right)\rho^{1/2}\\
    \label{FFdagger-action}
    &= \frac{1}{D} A\rho,
\end{align}
where in the last line we used the second identity in \eqref{characteristic}. Now, choose an orthonormal eigenbasis $\{\ket{\psi_j}:1\leq j\leq D\}$ of $\rho$, so that $\rho = \sum_{j=1}^D \lambda_j\ket{\psi_j}\!\bra{\psi_j}$. The operators
\begin{equation}
    E_{ij} := \sqrt{D}\,\ket{\psi_i}\!\bra{\psi_j}, \qquad 1 \leq i,j \leq D,
\end{equation}
form an orthonormal basis of $L^2(\Hn)$. By \eqref{FFdagger-action}, we get
\begin{equation}
    FF^\dagger(E_{ij}) = \frac{1}{D} E_{ij}\rho = \frac{\lambda_j}{D}E_{ij}.
\end{equation}
Thus the eigenvalues of $FF^\dagger$ are $\lambda_j/D$, with $D$ occurrences for each $j$. The operators $F^\dagger F$ and $FF^\dagger$ have the same
nonzero eigenvalues, including multiplicities. Moreover, $\dim\ell^2(\V) = N = L^2(\Hn)$, so their zero eigenvalues also have the same multiplicity. In view of \eqref{FdaggerF-kernel}, the eigenvalues of $N^{-1}K_\rho$ are therefore $\lambda_j/D$, each repeated $D$ times. They are nonnegative and sum to one, so $N^{-1}K_\rho$ is a state.

Finally,
\begin{align}
    S\left(\frac{1}{N} K_\rho\right) &= -\sum_{j=1}^D D\,\frac{\lambda_j}{D}\log\frac{\lambda_j}{D}\\
    &= -\sum_{j=1}^D \lambda_j\log\lambda_j + \log{D}\\
    &= S(\rho) + \log{D},
\end{align}
which prove \eqref{spectral-entropy}.
\end{proof}

%%%%%%%%%%%%%%%%%%%%%%%%%%%%%%%%%%%%%%%%%
%%%%%%%%%%%%%%%%%%%%%%%%%%%%%%%%%%%%%%%%%
%%%%%%%%%%%%%%%%%%%%%%%%%%%%%%%%%%%%%%%%%

\section{Submodularity of quantum convolution entropy}\label{main section}

We now fix inputs $(\rho_i)_{i\in\cI}$ and nonzero weights $(q_i)_{i\in\cI}$. All convolutions in this section belong to the compatible family of Definition~\ref{compatible}, unless a new choice of weights is explicitly made.

A real function $g$ on the subsets of a finite set is normalized if $g(\varnothing) = 0$, monotone if $g(A) \leq g(B)$ for $A \subseteq B$, and submodular if
\begin{equation}\label{submodular-definition}
    g(A) + g(B) \geq g(A\cap B) + g(A\cup B).
\end{equation}
A normalized monotone submodular function is also called a polymatroid rank function; integrality is not part of this convention. The (equivalent) diminishing-returns form of \eqref{submodular-definition} is
\begin{equation}\label{diminishing-definition}
    g(A\cup\{i\}) - g(A) \geq g(B\cup\{i\}) - g(B), \qquad A \subseteq B, \quad i\notin B.
\end{equation}

For a nonempty finite set $M$, a fractional cover consists of a family $\cF$ of nonempty subsets of $M$ and real weights $\alpha_J \geq 0$ such that
\begin{equation}
    \sum_{\substack{J\in\cF\\i\in J}} \alpha_J \geq 1, \qquad i\in M.
\end{equation}
Weights equal to zero may be omitted. Equality for every $i$ defines a fractional partition. The singleton cover has weight one on each singleton. The family of all $\ell$-element subsets of an $r$-element set is a fractional partition when every member is given weight $\binom{r-1}{\ell-1}^{-1}$.

\medskip
We are now ready to state and prove our main result.

\begin{theorem}\label{main theorem}
Let $R\in\Ad$ be nonempty. The function
\begin{equation}
    J \longmapsto S(\mathcal{C}_{R \cup J}) - S(\mathcal{C}_R), \qquad J\subseteq\cI\setminus R, \quad R \cup J \in \Ad,
\end{equation}
admits a normalized, monotone, submodular extension $g_R$ to $2^{\cI\setminus R}$. In particular,
\begin{equation}\label{anchored-physical}
    g_R(J) = S(\CA_{R\cup J}) - S(\CA_R), \qquad \text{whenever } R\cup J\in\Ad.
\end{equation}
Furthermore, if $(\alpha_J)_{J\in\cF}$ is a fractional cover of a nonempty $M \subseteq \cI\setminus R$, and if $R\cup J\in\Ad$ whenever $\alpha_J > 0$, then for $R\cup M\in\Ad$, we have
\begin{equation}\label{fractional-growth}
    S(\CA_{R\cup M}) - S(\CA_R) \leq \sum_{J\in\cF} \alpha_J \Big(S(\CA_{R\cup J})-S(\CA_R)\Big).
\end{equation}
No admissibility assumptions are required for subsets other than those explicitly appearing above.
\end{theorem}

\begin{remark}
The inequality in \eqref{fractional-growth} is a quantum extension of the fractional-cover entropic Pl\"unnecke–Ruzsa inequality of Madiman, Marcus, and Tetali \cite[Theorem 2.7]{MMT}. For diagonal inputs, every compatible convolution is an ordinary weighted sum of independent variables:
\begin{equation}\label{diagonal-specialization}
    \rho_i = \sum_z p_i(z)\proj{z} \quad \Longrightarrow \quad S(\CA_A) = H\left(c_A^{-1}\sum_{i\in A} q_iX_i\right) = H\left(\sum_{i\in A} q_iX_i\right).
\end{equation}
One may verify this directly from the characteristic function, or by evaluating the channel of Proposition~\ref{product theorem} on diagonal inputs. The singleton-cover and repeated-input bounds are direct entropy-growth counterparts of the Pl\"unnecke--Ruzsa principle developed in \cite{Tao,MMT}.
Their quantum extension concerns arbitrary density operators, including noncommuting inputs. The finite-field square condition is required for the quantum normalization used here; it is not a restriction on forming the corresponding classical weighted sums.
\end{remark}

\begin{proof}[Proof of Theorem~\ref{main theorem}]
As in Proposition~\ref{proposition spectrum}, for a state $\rho$, we define
\begin{equation}
    a_x^\rho := w(x)\rho^{1/2}\in L^2(\Hn), \qquad x\in\V.
\end{equation}
Then, by \eqref{kernel Hilbert Schmidt} we have $\langle a_x^\rho,a_y^\rho\rangle = \frac{1}{D} K_\rho(x,y)$. In particular, $\|a_x^\rho\|^2 = 1/D$, so that $\sqrt D\,a_x^\rho$ is a unit vector in $L^2(\Hn)$. Take one copy of $L^2(\Hn)$ for each input and define
the auxiliary state
\begin{equation}\label{omega}
    \Omega := \frac{1}{N} \sum_{x\in\V} \bigotimes_{i\in\cI} \proj{\sqrt{D}\,a_{q_i x}^{\rho_i}}.
\end{equation}
Thus $\Omega$ is a convex combination of product pure states on $\bigotimes_{i\in\cI} L^2(\Hn)$. Since each factor in \eqref{omega} is a rank-one state, tracing out registers simply removes the corresponding factors while retaining the common label $x$ in the remaining ones. Hence, for every nonempty $A\subseteq\cI$,
\begin{equation}
    \Omega_A = \frac{1}{N} \sum_{x\in\V} \bigotimes_{i\in A} \proj{\sqrt{D}\,a_{q_i x}^{\rho_i}}.
\end{equation}
Let $\{e_x:x\in\V\}$ denote the canonical orthonormal basis of $\ell^2(\V)$ and define
\begin{equation}
    F_A:\ell^2(\V) \longrightarrow \bigotimes_{i\in A} L^2(\Hn)
\end{equation}
by
\begin{equation}\label{FA}
    F_Ae_x := \frac{1}{\sqrt{N}} \bigotimes_{i\in A} \sqrt{D}\,a_{q_i x}^{\rho_i}.
\end{equation}
Equivalently, in Dirac notation for the Hilbert--Schmidt factors,
\begin{equation}
    F_A = \frac{1}{\sqrt{N}} \sum_{x\in\V} \left(\bigotimes_{i\in A} \ket{\sqrt{D}\,a_{q_i x}^{\rho_i}}\right)\bra{e_x}.
\end{equation}
It follows immediately from \eqref{omega} that
\begin{equation}
    F_AF_A^\dagger = \Omega_A.
\end{equation}
On the other hand, using \eqref{kernel Hilbert Schmidt},
\begin{align}
    \langle e_x,F_A^\dagger F_A e_y\rangle &= \langle F_A e_x,F_A e_y\rangle\\
    &= \frac{1}{N} \prod_{i\in A} \left\langle\sqrt{D}\,a_{q_i x}^{\rho_i},\sqrt{D}\,a_{q_i y}^{\rho_i}\right\rangle\\
    &= \frac{1}{N} \prod_{i\in A} K_{\rho_i}(q_i x,q_i y).
\end{align}
Therefore
\begin{equation}
    F_A^\dagger F_A = \frac{1}{N} \bigodot_{i\in A} K_{\rho_i}^{[q_i]}.
\end{equation}
Since $F_AF_A^\dagger$ and $F_A^\dagger F_A$ have the same nonzero eigenvalues, including multiplicities, and both have trace one, we obtain
\begin{equation}\label{marginalss}
    S(\Omega_A) = S\left(\frac{1}{N} \bigodot_{i\in A} K_{\rho_i}^{[q_i]}\right)
\end{equation}
for every nonempty $A\subseteq\cI$.

Now we define the auxiliary set function
\begin{equation}
    h(A) := S(\Omega_A), \qquad A\subseteq\mathcal I,
\end{equation}
with $h(\varnothing) = 0$. This set function is normalized, monotone, and submodular. To verify submodularity, let $A,B \subseteq \mathcal{I}$. When $A \cap B \neq \varnothing$, we apply strong subadditivity of $S$ to the pairwise disjoint register sets $A \setminus B$, $A \cap B$, and $B \setminus A$. This gives
\begin{equation}
    S(\Omega_A) + S(\Omega_B) \geq S(\Omega_{A \cap B}) + S(\Omega_{A \cup B}),
\end{equation}
or equivalently,
\begin{equation}
    h(A) + h(B) \geq h(A \cap B) + h(A \cup B).
\end{equation}
When $A \cap B = \varnothing$, the same inequality follows from ordinary subadditivity, since $h(\varnothing) = 0$. Finally, because $\Omega$ is separable, its conditional entropy is nonnegative across every bipartition. Hence, whenever $A \subseteq b$,
\begin{equation}
    h(B) - h(A) = S(B \setminus A \mid A)_\Omega \geq 0,
\end{equation}
which proves monotonicity.

Moreover, if $A\in\Ad$, then \eqref{marginalss} combined with Theorem~\ref{product theorem} and Proposition~\ref{proposition spectrum} gives
\begin{equation}\label{entropy-bridge}
 h(A)=S(\CA_A)+\log D.
\end{equation}
Next, we define the function $g_R$ on $2^{\mathcal{I}\setminus R}$ by
\begin{equation}
    g_R(J) := h(R\cup J) - h(R).
\end{equation}
this function is clearly normalized, and its monotonicity follows directly from the corresponding property of $h$. For $J,L\subseteq\cI\setminus R$, submodularity of $h$ applied to $R\cup J$ and $R\cup L$ gives
\begin{equation}
    g_R(J) + g_R(L) \geq g_R(J\cap L) + g_R(J\cup L).
\end{equation}
Equation~\eqref{anchored-physical} follows fro \eqref{entropy-bridge}.

Finally, we prove the fractional-cover inequality \eqref{fractional-growth} for $g_R$. Order $M = \{i_1,\ldots,i_r\}$ and put $P_j = \{i_1,\ldots,i_j\}$, with $P_0 = \varnothing$. Let
\begin{equation}
    d_j = g_R(P_j) - g_R(P_{j-1}) \geq 0.
\end{equation}
Expanding $g_R(J)$ in this order and using the diminishing-returns inequality \eqref{diminishing-definition} for $g_R$ gives
\begin{equation}
    g_R(J) = \sum_{j:i_j\in J} \bigl[g_R((J\cap P_{j-1})\cup\{i_j\}) - g_R(J\cap P_{j-1})\bigr] \geq \sum_{j:i_j\in J} d_j.
\end{equation}
Multiplying by $\alpha_J$, summing over $J$, and using the cover condition and $d_j \geq 0$, we obtain
\begin{equation}\label{fractional-abstract}
    \sum_{J\in\cF} \alpha_J g_R(J) \geq \sum_{j=1}^r d_j \sum_{J:i_j\in J} \alpha_J \geq \sum_{j=1}^r d_j = g_R(M).
\end{equation}
Translating this inequality for the displayed admissible subsets using \eqref{anchored-physical} proves \eqref{fractional-growth}. All intermediate quantities in the proof are marginals of $\Omega$ and exist even when they have no physical convolution interpretation.
\end{proof}

\begin{remark}
There is a general cancellation rule. If an entropy inequality for the marginals of \eqref{omega} has the form $\sum_A \gamma_A h(A) \geq 0$, all its nonzero terms are indexed by $A\in\Ad$, and $\sum_A \gamma_A = 0$, then
\begin{equation}\label{constant-cancellation}
    \sum_A \gamma_A S(\CA_A) \geq 0.
\end{equation}
The condition on the sum of the coefficients cancels the common $\log D$ shift. It should not be confused with balance of binary convolution coefficients or with the usual variable-by-variable notion of a balanced classical information inequality.
\end{remark}

%%%%%%%%%%%%%%%%%%%%%%%%%%%%%%%%%%%%%%%%%%%%
%%%%%%%%%%%%%%%%%%%%%%%%%%%%%%%%%%%%%%%%%%%%
%%%%%%%%%%%%%%%%%%%%%%%%%%%%%%%%%%%%%%%%%%%%

\section{Entropy inequalities from submodularity}\label{corollary section}

\subsection{Convolutional strong subadditivity}

\begin{corollary}[Compatible submodularity]\label{corollary submodularity}
If $A\subseteq B$ and $A,B\in\Ad$, then $S(\CA_A) \leq S(\CA_B)$. If $A,B,A\cap B,A\cup B\in\Ad$, with $A\cap B \neq \varnothing$, then
\begin{equation}\label{physical-submodularity}
    S(\CA_A) + S(\CA_B) \geq S(\CA_{A\cap B}) + S(\CA_{A\cup B}).
\end{equation}
Its deficit is exactly $I(A\setminus B:B\setminus A\mid A\cap B)_\Omega = S(\CA_A) + S(\CA_B) - S(\CA_{A\cap B}) - S(\CA_{A\cup B})$.
\end{corollary}
\begin{proof}
Apply monotonicity and submodularity to $h$, and use \eqref{entropy-bridge}. For the last assertion, expand the conditional mutual information.
\end{proof}

Thus, whenever the four displayed subsets are admissible,
\begin{equation}\label{physical-diminishing}
    S(\CA_{K\cup\{i\}}) - S(\CA_K) \geq S(\CA_{K\cup\{i,j\}}) - S(\CA_{K\cup\{j\}}),
\end{equation}
for nonempty $K$ and distinct $i,j\notin K$. This is the diminishing-returns law for globally weighted convolution entropy. For instance, four inputs give
\begin{equation}\label{four-overlap}
    S(\CA_{123}) + S(\CA_{234}) \geq S(\CA_{23}) + S(\CA_{1234}),
\end{equation}
where a string of indices denotes the corresponding subset.

For nested binary convolutions, we have the following

\begin{corollary}[Weighted convolutional strong subadditivity]\label{corollary weighted-ssa}
Let $(a,b)$, $(l,m)$, and $(s,t)$ be admissible binary pairs, with all six coefficients nonzero, such that $l b t = m s$. Then, for arbitrary states $\rho,\sigma,\tau$,
\begin{equation}\label{weighted-ssa}
    S\big((\rho\conv_{a,b}\sigma)\conv_{l,m}\tau\big) + S(\sigma) \leq S(\rho\conv_{a,b}\sigma) + S(\sigma\conv_{s,t}\tau).
\end{equation}
\end{corollary}
\begin{proof}
Assign weights $(la,lb,m)$ to $(\rho,\sigma,\tau)$. Their squared sum is $l^2(a^2 + b^2) + m^2 = 1$. For the first pair take root $l$, so its normalized convolution is $\rho\conv_{a,b}\sigma$. Since $l b t = m s$, the nonzero scalar $r = lb/s = m/t$ satisfies $(lb)^2 + m^2 = r^2(s^2 + t^2) = r^2$. The second pair is therefore admissible with root $r$ and normalized convolution $\sigma\conv_{s,t}\tau$. The full convolution has characteristic function $\Xi_\rho(la x)\Xi_\sigma(lb x)\Xi_\tau(mx)$, which is that of $(\rho\conv_{a,b}\sigma)\conv_{l,m}\tau$. Apply \eqref{physical-submodularity} to the two pairs, whose intersection is the register of $\sigma$.
\end{proof}

Condition $l b t = m s$ is a sufficient compatibility condition for this theorem. We do not assert that it is necessary for every entropy inequality of the form \eqref{weighted-ssa}. For a concrete unequal-coefficient example, in $\Z_{23}$ the pairs $(a,b) = (4,10)$, $(l,m) = (8,11)$, and $(s,t) = (14,14)$ are admissible and satisfy $l b t = m s$. Thus the weighted statement includes parameter choices beyond the balanced specialization.

\begin{corollary}[Convolutional strong subadditivity]\label{corollary balanced-ssa}
Suppose $s,l,m\in\Z_d^\times$ satisfy
\begin{equation}\label{balanced-three-parameters}
    2s^2 = 1, \qquad l^2 + m^2 = 1, \qquad m = ls.
\end{equation}
Then for all states $\rho,\sigma,\tau$,
\begin{equation}\label{balanced-ssa}
    S\big((\rho\conv_{s,s}\tau)\conv_{l,m}\sigma\big) + S(\sigma) \leq S(\rho\conv_{s,s}\sigma) + S(\sigma\conv_{s,s}\tau).
\end{equation}
\end{corollary}

\begin{proof}
Set $a = b = s = t$ in Corollary~\ref{corollary weighted-ssa}. The compatibility condition becomes $ls^2 = ms$, equivalent to $ls=m$. The characteristic coefficients in the full convolution are then $(ls,ls,m) = (m,m,m)$, so its output is invariant under permuting the three inputs. This gives the nesting displayed in \eqref{balanced-ssa}.
\end{proof}

This proves \cite[Conjecture 2]{BGJRuzsa}. The coefficient assumptions are equivalent to the existence of nonzero square roots of both $2$ and $3$ in $\Z_d$: indeed $3m^2 = 1$ and $l = m/s$. They will not be needed in full for the triangle inequality below.

\subsection{Quantum Ruzsa triangle inequality}

For any fixed admissible convolution $\conv=\conv_{s,t}$, the quantum Ruzsa divergence and its symmetrized version are
\begin{align}
    D_{\Rz}(\rho\Vert\sigma) &:= S(\rho\conv\sigma)-S(\rho),\label{ruzsa-general-definition}\\
    d_{\Rz}(\rho,\sigma) &:= \tfrac{1}{2}\big[S(\rho\conv\sigma) + S(\sigma\conv\rho) - S(\rho) - S(\sigma)\big].\label{ruzsa-symmetric-definition}
\end{align}
These are Definitions~23 and~24 of \cite{BGJRuzsa}. The
convolution need not be balanced for these definitions. In the
balanced case, commutativity reduces the first two entropies in
\eqref{ruzsa-symmetric-definition} to the same term. These symbols always refer to the pair chosen in that context.

Fix $s\in\Z_d^\times$ with $2s^2 = 1$. With $\conv = \conv_{s,s}$, the definition \eqref{ruzsa-general-definition} becomes
\begin{equation}\label{ruzsa-definition}
    D_{\Rz}(\rho\Vert\sigma) = S(\rho\conv_{s,s}\sigma) - S(\rho).
\end{equation}
The quantity is nonnegative by Corollary~\ref{corollary submodularity}. It is a divergence, not a metric: it is generally asymmetric and need not vanish on the diagonal.

\begin{corollary}[Quantum Ruzsa triangle inequality]\label{corollary ruzsa}
Under the binary condition $2s^2 = 1$, arbitrary states satisfy
\begin{equation}\label{ruzsa-triangle}
    D_{\Rz}(\rho\Vert\tau) \leq D_{\Rz}(\rho\Vert\sigma) + D_{\Rz}(\sigma\Vert\tau).
\end{equation}
Equivalently,
\begin{equation}\label{entropy-triangle}
    S(\rho\conv_{s,s}\tau) + S(\sigma) \leq S(\rho\conv_{s,s}\sigma) + S(\sigma\conv_{s,s}\tau).
\end{equation}
No balanced three-input normalization is assumed.
\end{corollary}
\begin{proof}
Give the three inputs weights $(1,1,1)$ and label their auxiliary registers $A,B,C$. Every pair is admissible with root $s^{-1}$. Using the auxiliary entropy function $h$ constructed in the proof of Theorem~\ref{main theorem}, submodularity and monotonicity give
\begin{equation}
    h(AB) + h(BC) \geq h(B) + h(ABC) \geq h(B) + h(AC).
\end{equation}
The entropy bridge \eqref{entropy-bridge} now converts this inequality into \eqref{entropy-triangle}. Notice that $ABC$ is used only as an auxiliary marginal and therefore need not be admissible. Substituting \eqref{ruzsa-definition} yields \eqref{ruzsa-triangle}.
\end{proof}

Thus \eqref{ruzsa-triangle} proves \cite[Conjecture 1]{BGJRuzsa} whenever the balanced binary convolution exists, requiring only a solution of $2s^2=1$. No balanced three-input normalization is needed. This distinction is substantive: over $\Z_7$, the choice $s=2$ satisfies $2s^2 = 1$, whereas the additional conditions $m = ls$ and $l^2 + m^2 = 1$ in Corollary~\ref{corollary balanced-ssa} would imply $l^2 = 3$, which is impossible because $3$ is not a quadratic residue modulo $7$. Hence, over $\Z_7$, \eqref{ruzsa-triangle} holds although the convolutional strong subadditivity (Corollary~\ref{corollary balanced-ssa}) cannot even be instantiated.

The symmetrized divergence of \eqref{ruzsa-symmetric-definition} has the balanced form
\begin{equation}\label{symmetric-ruzsa}
    d_{\Rz}(\rho,\sigma) = S(\rho\conv_{s,s}\sigma) - \tfrac{1}{2} S(\rho) - \tfrac{2}{2} S(\sigma).
\end{equation}
It satisfies the same triangle inequality, since $d_{\Rz}(\rho,\sigma) + d_{\Rz}(\sigma,\tau) - d_{\Rz}(\rho,\tau)$ is the deficit in \eqref{entropy-triangle}. It is also nonnegative, but its diagonal value is generally positive. Convolution-based symmetric divergences and triangle questions
are part of the related work in \cite{Hou}.

\begin{remark}
For diagonal states with corresponding independent random variables $X,Y,Z$, the balanced convolution has the law of $s(X + Y)$. Since multiplication by $s \neq 0$ is a bijection of $\Z_d^n$, it preserves Shannon entropy, and \eqref{entropy-triangle} reduces to
\begin{equation}
    H(X+Z) + H(Y) \leq H(X + Y) + H(Y + Z).
\end{equation}
This is a triangle inequality for the divergence used in \eqref{ruzsa-definition}. It should be distinguished from Tao's
entropic Ruzsa distance \cite{Tao,GMT}
\begin{equation}
    d_{\Rz}(X,Y) = H(X - Y) - \frac{1}{2} H(X) - \frac{1}{2} H(Y),
\end{equation}
whose triangle inequality is equivalent to
\begin{equation}
    H(X - Z) + H(Y) \leq H(X - Y) + H(Y - Z).
\end{equation}
The distinction lies in the sign pattern. Our globally weighted construction assigns one fixed coefficient to each input across all subsets, whereas the three difference terms above cannot be produced simultaneously from a single assignment of signs to $X,Y,Z$. Thus negative global weights are permitted, but the main theorem does not automatically recover Tao's difference-based triangle inequality.
\end{remark}

\subsection{Quantum Pl\"unnecke--Ruzsa inequalities}

Theorem~\ref{main theorem} is a quantum version of the direct fractional-cover entropy bounds for sums in \cite{MMT}. As in the classical theory, small entropy growth against one fixed input controls growth against many inputs. The bounds hold directly for the original input states; no analogue of passing to a favorable subset of the base object is required.

For the singleton cover, it takes the particularly transparent form
\begin{equation}\label{singleton-growth}
    S(\CA_{R\cup M}) - S(\CA_R) \leq \sum_{j\in M} S(\CA_{R\cup\{j\}}) - S(\CA_R).
\end{equation}
In terms of entropy growth factors $\kappa_R(J) = \exp(S(\CA_{R\cup J}) - S(\CA_R))$, the full fractional-cover version is
\begin{equation}\label{growth-factors}
    \kappa_R(M) \leq \prod_{J\in\cF} \kappa_R(J)^{\alpha_J}.
\end{equation}
These factors play the role of multiplicative sumset growth. For example, if $|M|=r$ and each singleton entropy gain is at most $\log K$, then \eqref{singleton-growth} gives $\kappa_R(M) \leq K^r$, subject only to the displayed endpoint normalizations.

With base register $0$ and added registers $1,2,3$, the singleton cover and the cover by pairs of weight $1/2$ respectively give
\begin{align}
    S(\CA_{0123}) + 2S(\rho_0) &\leq S(\CA_{01}) + S(\CA_{02}) + S(\CA_{03}),\label{four-star}\\
    2S(\CA_{0123}) + S(\rho_0) &\leq S(\CA_{012}) + S(\CA_{013}) + S(\CA_{023}).\label{four-fractional}
\end{align}
Each inequality requires exactly the subsets displayed in it to be admissible. In particular, \eqref{four-star} requires no triple-subset normalization.

\begin{corollary}[Balanced four-input growth]\label{corollary four-input}
If $2s^2=1$, then
\begin{equation}\label{four-balanced}
    S\big((\rho_0\conv_{s,s}\rho_1) \conv_{s,s} (\rho_2\conv_{s,s}\rho_3)\big) + 2S(\rho_0) \leq \sum_{j=1}^3 S(\rho_0\conv_{s,s}\rho_j).
\end{equation}
\end{corollary}
\begin{proof}
Take all four global weights equal to one. Every pair is admissible, and the full set is admissible with root $2$. Its normalized characteristic coefficients are all $1/2 = s^2$, which are exactly the effective coefficients in the four-input convolution. Apply \eqref{four-star}.
\end{proof}

\begin{remark}
The preceding inequality applies over $\Z_7$, where one may take $s=2$. An attempt to derive it by iterating three-input convolution inequalities would require intermediate equal-weight convolutions of three inputs. Such a convolution would have a common coefficient $t$ satisfying $3t^2 = 1$, which has no solution in $\Z_7$. Our fractional-cover argument avoids this obstruction because no intermediate triple convolution is required.
\end{remark}

\subsection{Repeated inputs and optimal coefficients}\label{repeated inputs and optimal coefficients}

The fractional partitions by equally sized subsets yield a hierarchy between different repetition scales. It is useful first to allow a distinguished base state.

\begin{corollary}[Entropy gain per added copy]\label{corollary copies}
Let $a,b\in\Z_d^\times$, and fix states $\rho,\sigma$. For an integer $j \geq 1$ with $a^2 + j b^2 = c_j^2 \neq 0$, $c_j\in\Z_d^\times$, let $T_j$ be
the state with characteristic function
\begin{equation}\label{repeated-characteristic}
    \Xi_{T_j}(x) = \Xi_\rho\left(\frac{a}{c_j}x\right)\Xi_\sigma\left(\frac{b}{c_j}x\right)^j.
\end{equation}
Put $T_0 = \rho$. If $1 \leq \ell \leq k$ and both $T_\ell,T_k$ are defined, then
\begin{equation}\label{copy-comparison}
    \frac{S(T_k) - S(\rho)}{k} \leq \frac{S(T_\ell) - S(\rho)}{\ell}.
\end{equation}
No normalization at intermediate repetition counts is required.
\end{corollary}
\begin{proof}
Take one base input $\rho$ of weight $a$ and $k$ independent copies of $\sigma$, each of weight $b$. Cover the copy indices by all their $\ell$-element subsets, each of weight $\binom{k - 1}{\ell - 1}^{-1}$. All those subsets together with the base have the same convolution entropy $S(T_\ell)$. Thus
Theorem~\ref{main theorem} gives
\begin{equation}
    S(T_k) - S(\rho) \leq \frac{\binom{k}{\ell}}{\binom{k - 1}{\ell - 1}}[S(T_\ell) - S(\rho)] = \frac{k}{\ell}[S(T_\ell) - S(\rho)],
\end{equation}
as desired.
\end{proof}

For a single input state, write $\mathcal{C}_{[m]}(\rho)$ for its balanced $m$-input convolution whenever $m$ is a nonzero square in $\Z_d$. Thus, for either choice of $t_m$ with $m t_m^2 = 1$,
\begin{equation}\label{balanced-m-characteristic}
    \Xi_{\mathcal{C}_{[m]}(\rho)}(x) = \Xi_\rho(t_m x)^m, \qquad \mathcal{C}_{[1]}(\rho) = \rho
\end{equation}
with $t_1 = 1$. Entropy is independent of the root chosen.

The quantum doubling constant \cite[Definition 42]{BGJRuzsa} is defined as
\begin{equation}\label{doubling-general-definition}
    \delta_q[\rho] := \exp\big(S(\rho\conv\rho) - S(\rho)\big) = \exp\big(D_{\Rz}(\rho\Vert\rho)\big).
\end{equation}

\begin{corollary}[Balanced repetition hierarchy]\label{corollary balanced-growth}
For $2 \leq \ell \leq m$ such that $\ell$ and $m$ are nonzero squares in $\Z_d$,
\begin{equation}\label{balanced-growth}
    \frac{S(\mathcal{C}_{[m]}(\rho)) - S(\rho)}{m-1} \leq \frac{S(\mathcal{C}_{[\ell]}(\rho)) - S(\rho)}{\ell-1}.
\end{equation}
If balanced binary convolution exists, then every admissible $m\ge2$ satisfies
\begin{equation}\label{doubling-growth}
    S(\mathcal{C}_{[m]}(\rho)) - S(\rho) \leq (m-1) \log\delta_q[\rho].
\end{equation}
\end{corollary}

\begin{proof}
Apply Corollary~\ref{corollary copies} with $a = b = 1$, $\sigma = \rho$, and copy counts $\ell - 1$ and $m - 1$. Then set $\ell = 2$.
\end{proof}

The integers $\ell,m$ in entropy ratios and exponents are ordinary real numbers; only the normalization equations use their residues in $\Z_d$. If $2s^2 = 1$, every $m = 2^r$ is admissible, with $t_m = s^r$. A complete balanced binary tree of depth $r$ has that characteristic function. For other grouping schemes, coefficients must be tracked using \eqref{grouping}; the fixed-pair iterate $\conv_{s,t}^{m - 1}\rho$ in \eqref{bgj-iteration} generally has different global weights and is not being identified with $\mathcal{C}_{[m]}(\rho)$. In particular, \eqref{doubling-growth} is a bound for the normalized convolutions \eqref{balanced-m-characteristic}.

\begin{proposition}[Sharpness]\label{proposition sharpness}
For every fixed odd prime $d$ and fixed admissible integers $2 \leq \ell \leq m$, the factor $(m - 1)/(\ell - 1)$ in \eqref{balanced-growth} cannot be replaced by a smaller universal factor, even for one-qudit diagonal states. More generally, for fixed $a,b\ne0$ and admissible counts $1 \leq \ell \leq k$, the factor $k/\ell$ in the equivalent unnormalized form of \eqref{copy-comparison} is optimal.
\end{proposition}

\begin{proof}
Let
\begin{equation}
        \rho_\varepsilon = (1 - \varepsilon)\proj{0} + \varepsilon\proj{1}, \qquad 0 < \varepsilon < 1.
\end{equation}
For admissible $j$, the diagonal distribution of $\mathcal{C}_{[j]}(\rho_\varepsilon)$ is the law of $t_j(X_1+\cdots+X_j)$ in $\Z_d$, where the $X_i$
are independent Bernoulli variables of parameter $\varepsilon$. Multiplication by $t_j\ne0$ preserves entropy. The event of no success has probability $1 - j\varepsilon + O_j(\varepsilon^2)$, the event of exactly one success has probability $j\varepsilon + O_j(\varepsilon^2)$, and all other events together have probability $O_j(\varepsilon^2)$. Reduction modulo $d$ does not identify the residues $0$ and $1$; contributions from two or more successes affect only the error terms. It follows that
\begin{equation}
        S(\mathcal{C}_{[j]}(\rho_\varepsilon)) = j\varepsilon\log(1/\varepsilon) + O_{d,j}(\varepsilon), \qquad S(\rho_\varepsilon) = \varepsilon\log(1/\varepsilon) + O(\varepsilon).
\end{equation}
Consequently,
\begin{equation}
        \lim_{\varepsilon \downarrow 0}\frac{S(\mathcal{C}_{[m]}(\rho_\varepsilon)) - S(\rho_\varepsilon)}{S(\mathcal{C}_{[\ell]}(\rho_\varepsilon)) - S(\rho_\varepsilon)} = \frac{m - 1}{\ell - 1}.
\end{equation}
For $\ell > 1$ the denominator is positive for sufficiently small $\varepsilon$, which proves the assertion.

For the more general statement, take $\rho = \sigma = \rho_\varepsilon$ in \eqref{repeated-characteristic}. The diagonal law of $T_j$ is that of $c_j^{-1}(aX_0 + bX_1 +\cdots+ bX_j)$. Each event of exactly one success has a nonzero residue because $a,b \neq 0$. Whether these residues coincide or are distinct, their total first-order mass is $(j + 1)\varepsilon$. Consequently $S(T_j) = (j + 1)\varepsilon\log(1/\varepsilon) + O_{d,j}(\varepsilon)$, and the ratio of the gains at counts $k$ and $\ell$ tends to $k/\ell$.
\end{proof}

%%%%%%%%%%%%%%%%%%%%%%%%%%%%%%%%%%%%%%%%%%
%%%%%%%%%%%%%%%%%%%%%%%%%%%%%%%%%%%%%%%%%%
%%%%%%%%%%%%%%%%%%%%%%%%%%%%%%%%%%%%%%%%%%

\section{Conclusion}\label{conclusion section}

The characteristic-kernel representation places all convolutions in a common entropy space. Three features make this possible: the inputs occupy independent registers, the convolution weights are inherited consistently across subsets, and every admissible subcollection admits the required coefficient normalization. These ingredients produce a separable auxiliary state whose marginal entropies reproduce the physical convolution entropies up to explicit constants. Ordinary quantum strong subadditivity then yields, relative to every fixed nonempty admissible input block, a normalized, monotone, submodular extension of the physical entropy gains. This single mechanism accounts for convolutional strong subadditivity, the quantum Ruzsa triangle inequality, and the direct Pl\"unnecke--Ruzsa hierarchy.

The resulting theory lies on the direct-growth side of entropic additive combinatorics. On the classical inverse side, the recent proofs of Marton's polynomial Freiman--Ruzsa conjecture in characteristic two and, more generally, in abelian groups of bounded torsion show that small doubling forces polynomially controlled algebraic structure \cite{GGMT,GGMTTorsion}. Building on these results, a polynomial-time algorithmic Freiman--Ruzsa theorem for subsets of $\mathbb F_2^n$ was recently obtained, with applications to stabilizer-state tomography and the learning of quantum states with bounded stabilizer extent \cite{ADGG}.

These inverse developments exploit additional additive structure that is not captured by the characteristic-kernel method. Understanding whether analogous structure can be extracted from small quantum convolutional entropy growth, and how it relates to stabilizer structure and magic, is a natural direction for further work.

\paragraph{Acknowledgements.} The research is supported by the National Research Foundation (NRF-NRFI10-2024-0006), Singapore, through the National Quantum Office, hosted by A*STAR, under its Centre for Quantum Technologies Funding Initiative (S24Q2d0009).

\end{document}